\documentclass[12pt]{article} %showpacs 
\usepackage[paper=a4paper,margin=1.0in]{geometry}

\usepackage[english]{babel}
\usepackage{amsthm}
\usepackage{soul}
\usepackage{amssymb,amstext,amsmath,amsthm,mathtools}
\usepackage[dvips]{graphicx}
\usepackage{latexsym}
\usepackage{psfrag}
\usepackage{amsfonts}
\usepackage{bm} % bold math
\usepackage{color}
\usepackage{cite}
\usepackage{hyperref}
\hypersetup{colorlinks,allcolors=blue,citecolor=red,linktocpage=true}
\usepackage[utf8]{inputenc} 
\usepackage{enumitem}
\numberwithin{equation}{section}
\usepackage{authblk}

\usepackage[normalem]{ulem}
\newcommand{\normord}[1]{:\mathrel{#1}:}

\newcommand{\hide}[1]{}

\def\bra{\langle}
\def\ket{\rangle}

\def\VPD{{\mbox{\tiny VPD}}}
\def\Diff{{\mbox{\tiny Diff}}}
\def\DW{{\mbox{\tiny Diff-Weyl}}}

\newtheorem{theorem}{Theorem}[section]

\begin{document}

\title{General Actions of Extended Objects \\ and Volume-Preserving Diffeomorphism}

\author{
\vskip1em
{\large
        Pei-Ming Ho$^{a,b}$\footnote{pmho@ntu.edu.tw},~
        Hikaru Kawai$^{a,b,c}$\footnote{hikarukawai@phys.ntu.edu.tw},~
        and Henry Liao$^a$\footnote{henryliao.physics@gmail.com}
 }
\\ \vskip0.2em
    {
    \small ${}^a$\emph{Department of Physics and Center for Theoretical Physics,\\}
    \vskip-0.1em\small\emph{National Taiwan University, Taipei 106, Taiwan}
    }
    
    {
    \small ${}^b$\emph{Physics Division, National Center for Theoretical Sciences, Taipei 10617, Taiwan}
    }
\\    
    {\vskip0.35em
        \small ${}^c$\emph{Nambu Yoichiro Institute of Theoretical and Experimental Physics (NITEP),} \\
     \vskip-0.05em{\small \emph{Osaka Metropolitan University, Osaka 558-8585, Japan}}
    }
\\
}

\date{}

\maketitle

\begin{abstract}
We consider actions that are general functions of the worldsheet/worldvolume metric and the induced metric for extended objects embedded in spacetime as Riemannian manifolds, areal-metric manifolds, and volume-metric manifolds.
For strings on a Riemannian spacetime, we consider general actions respecting volume-preserving diffeomorphisms (VPD), general diffeomorphisms, and diffeomorphisms with Weyl symmetry, respectively. 
Well-known Schild, Nambu-Goto, and Polyakov actions are included as special cases.
We reach two main conclusions:
(1) When actions are functions of both the worldsheet metric and induced metrics, all nontrivial self-consistent actions are classically equivalent.
(2) As a physical constraint on the classical action, VPD symmetry is as strong as the full diffeomorphism symmetry.
The discussion is then extended to strings in spacetime manifolds equipped with the areal or volume metrics.
Then, we further consider higher-dimensional extended objects in spacetime defined with areal or volume metrics, and show the equivalence between the generalized Schild actions and the generalized Nambu-Goto action.
We prove a general theorem on VPD that explains this equivalence.
Incidentally, while only the areal metric is needed to define the string worldsheet action, we show that the Polyakov action with an areal-metric perturbation cannot describe critical strings without other interaction terms.
\end{abstract}

%%%%%%%%%%%%%%%%%%%%%%%%%%%%%%%%%%%%%%%%%%%%%%%%%%%%%%%%%%%%%%%%%%%%%%%%%%%%%%%%%%%%
\section{Introduction}\label{s:intro}

At the classical level, the worldsheet theory of a bosonic string can be equivalently described by the Nambu-Goto (NG) action, the Schild (S) action, and the Polyakov (P) action as (with the overall coefficients omitted)
\begin{align}
    S_{NG} &\equiv \int d^2\xi \left(\sqrt{\gamma} \right),
    \label{SNG} \\
    S_{S} &\equiv \int d^2\xi \left( \frac{1}{4} \epsilon^{ac} \epsilon^{bd} \gamma_{ab} \gamma_{cd} \right) = \int d^2\xi \left(\frac{1}{2} \gamma\right), 
    \label{SS} \\
    S_{P} &\equiv \int d^2\xi \left( \frac{1}{2} \sqrt{h} h^{ab} \gamma_{ab} \right),
    \label{SP}
\end{align}
where the $\xi^a$ ($a, b, c, d = 0, 1$) are the worldsheet coordinates, and $\epsilon^{ab}$ is the two-dimensional Levi-Civita symbol.
The induced metric $\gamma_{ab}$ is defined by
\begin{align}
    \gamma_{ab} \equiv g_{\mu\nu}(X) \partial_a X^\mu \partial_b X^\nu,
    \label{eq:gamma}
\end{align} 
where $g_{\mu\nu}(X)$ is a generic spacetime metric, whose (absolute) determinant is denoted by $\gamma$, and $h_{ab}$ is the worldsheet metric with $h$ being (the absolute value of) its determinant.

The Nambu-Goto action~\eqref{SNG} is invariant under the worldsheet diffeomorphisms.
The Schild action is invariant only under the worldsheet volume-preserving diffeomorphisms.
The Polyakov action enjoys both the full worldsheet diffeomorphism and the Weyl symmetry that scales the worldsheet metric $h_{ab}$.
The conformal gauge allows us to define the quantum theory of strings.
The classical equivalence of all three actions can be easily shown~\cite{Polyakov:1981rd,the Schild:1976vq}.
This raises a natural question: {\em Are there other actions classically equivalent to these three?}

In this paper, we study the most general forms of the string worldsheet action for each choice of the worldsheet gauge symmetry mentioned above.
We discuss conditions under which these actions are classically equivalent.

The above actions are related to the area of the string worldsheet induced from the background (pseudo-)Riemannian metric.
It is natural to generalize the notion of a (pseudo-)Riemannian metric to that of an areal metric.
We shall also study general forms of the string worldsheet actions with this extension and examine their equivalence.

Let us briefly recall the notion of an areal metric.
As the relativistic action of a point particle is the mass times the worldline length, the Nambu-Goto action of a string is the tension times the worldsheet area.
On a (pseudo-)Riemannian manifold, we define length using a metric $g_{\mu\nu}$ via the line element
\begin{align} \label{eq:line_metric}
    ds^2 = g_{\mu\nu} dX^\mu \otimes_{sym} dX^\nu,
\end{align}
where $dX^\mu$ is a one-form on the manifold and $\otimes_{sym}$ is the symmetrized tensor product.
The area can then be defined as a notion induced from that of the length.

Mathematically, one can define the areal element $dA$ directly as an inner product on the space of two-forms~\cite{Schuller:2005yt,Schuller:2005ru,Punzi:2006hy,Punzi:2006nx,Ho:2015cza,Borissova:2023yxs,Borissova:2024cpx}
as 
\begin{align} \label{eq:areal_metric_def}
    dA^2 = G_{\mu\nu\rho\lambda} \left(dX^{\mu}\wedge dX^{\nu}\right) \otimes_{sym} \left(dX^{\rho}\wedge dX^{\lambda}\right),
\end{align}
where $G_{\mu\nu\rho\lambda}$ is called the areal metric.
An areal metric may or may not induce the notion of length~\cite{Ho:2015cza,Borissova:2024cpx}.
However, one can define the string worldsheet area even in the absence of the notion of length.

In this work, we also take the first step toward the quantum theory of strings in an areal metric background.
We consider the ordinary Polyakov action perturbed by an areal-metric deformation.
Unfortunately, we find that such perturbations do not define critical strings.

Beyond discussion of strings, we shall also consider generic extended objects with $d$-dimensional worldvolume as generalizations of strings, in a background geometry defined with a volume metric, as a generalization of the areal metric.

In the discussions on the classical equivalence among different actions, we notice that VPD plays a special role.
We found an interesting general theorem that is very useful in these discussions.

The plan of this paper is as follows.
In Secs.~\ref{Riemann} and~\ref{sec:the Nambu-Goto_the_Schild_equiv}, we consider the most general worldsheet actions that respect the following three types of symmetries: (1) volume-preserving diffeomorphisms (VPD), (2) general diffeomorphisms (Diff), and (3) diffeomorphism with Weyl symmetry (Diff-Weyl).
We show that they are all classically equivalent.
They can be viewed, respectively, as generalizations of (1) the Schild action~\eqref{SS}, (2) the Nambu-Goto action~\eqref{SNG}, and (3) the Polyakov action~\eqref{SP}.
In Sec.~\ref{sec:algebraic_areal_bosonic}, we replace the Riemannian metric by the areal metric in the Nambu-Goto, the Schild, and the Polyakov actions and prove their classical equivalence.
In Sec.~\ref{sec:derivation_of_general_lagrangians}, we consider the most general actions for these 3 types of symmetries in areal-metric backgrounds.
In Sec.~\ref{sec:anomalous_polyakov}, we study the quantum theory of the Polyakov action perturbed by the areal metric, which was proposed in~\cite{Borissova:2024cpx}.
We checked that this perturbation does not define a critical string.
It is therefore not clear whether string theories can be consistently defined in areal-metric backgrounds.
In Sec.~\ref{sec:higher_dim}, we provide some comments on a similar setup for higher-dimensional objects.
Specifically, we consider actions that are functions of the worldvolume metric and the induced metric on both Riemannian and volume-metric manifolds.
At the end of the section, we provide a theorem on VPD that can be applied to prove the classical equivalence between the Schild(-like) action and the Nambu-Goto action in arbitrary dimensions on both Riemannian and volume-metric manifolds.
The conclusion is given in Sec.~\ref{sec:conclusion}.

\section{Worldsheet Actions on Riemannian Manifolds}
\label{Riemann}

In this section, we aim to find the most general string worldsheet actions that depend only on the induced metric $\gamma_{ab}$ and the worldsheet metric $h_{ab}$.
That is, we assume that the worldsheet fields $X^{\mu}$ only enter the action through $\gamma_{ab}$~\eqref{eq:gamma}; the only other field that can appear in the action is the worldsheet metric $h_{ab}$, which plays the role of an auxiliary field.\footnote{Fermionic fields are ignored in this work.}
This assumption already covers all three types of string actions~\eqref{SNG} --~\eqref{SP}.

In Sec.~\ref{sec:general_worldsheet_actions}, we first discuss the basic idea to construct the general form of Lagrangians which are functions of $\gamma_{ab}$ and $h_{ab}$ without their derivatives. 
Then, we consider worldsheet Lagrangians that respect different types of symmetries: (1) Volume-Preserving Diffeomorphisms (VPD), (2) Diffeomorphism (Diff), and (3) Diffeomorphism with Weyl symmetry (Diff-Weyl).

After integrating out the auxiliary worldsheet metric $h_{ab}$, we find that the latter two classes of Lagrangians reduce to the Nambu-Goto Lagrangian, while Lagrangian with the VPD symmetry reduces to Lagrangians of the form
\begin{align}\label{eq:generalized_schlid_action_opening_2}
    \mathcal{L}_{GS}=F(\sqrt{\gamma}),
\end{align}
where $F$ is an arbitrary function, and we call it the generalized Schild Lagrangian.
The classical equivalence between the generalized Schild actions and the Nambu-Goto action will be shown later in Sec.~\ref{sec:the Nambu-Goto_the_Schild_equiv}.

\subsection{General Worldsheet Actions}\label{sec:general_worldsheet_actions}

Denoting the matrix form of $\gamma_{ab}$ as $\Gamma$ and the matrix form of $h_{ab}$ as $H$, invariance under VPD demands that the worldsheet Lagrangian is a function of the following invariants $\gamma = \det \Gamma$, $h = \det H$, $a_n = \text{tr}\left(\left(\Gamma H^{-1}\right)^n\right)$ ($n=\pm1,\pm2,\dots$).
Assuming that the Lagrangian does not contain derivatives of $\gamma_{ab}$ and $h_{ab}$, a Lagrangian invariant under the VPD symmetry is in general of the form $f(\gamma, h, a_1,\dots)$.

For invariants constructed from $\Gamma H^{-1}$, since $\Gamma H^{-1}$ is a $2\times 2$ matrix, it is characterized by two independent invariants (e.g. trace and determinant).
All other invariants $a_n$ can be identified with functions of them.
For example, take 
\begin{align}\label{eq:a1}
    a_1 &= \text{tr}\left(\Gamma H^{-1}\right) = \gamma_{ab} h^{ba},
    \\
    a_2 &= \text{tr}\left(\left(\Gamma H^{-1}\right)^2\right) = \gamma_{ab} h^{bc} \gamma_{cd} h^{da} \label{eq:a2}
\end{align} 
as the only independent variables for all invariants constructed from $\Gamma H^{-1}$. 
We have 
\begin{align}\label{eq:a1a2_l1l2}
    a_1 &= \lambda_1 + \lambda_2,
    \\
    a_2 &= (\lambda_1)^2 + (\lambda_2)^2,
\end{align}
where $\lambda_1$ and $\lambda_2$ are the two eigenvalues of $\Gamma H^{-1}$.
Then, $a_{-1} = \frac{1}{\lambda_1} + \frac{1}{\lambda_2}$ can be written in terms of $a_1$ and $a_2$ as $\frac{2 a_1}{a_2-(a_1)^2}$.

Once we have the characteristics of $\Gamma H^{-1}$, the determinants $\gamma$ and $h$ are no longer independent.
Indeed, we have
\begin{align}\label{gh-1}
\gamma h^{-1} = \lambda_1\lambda_2 = \frac{{a_1}^2-a_2}{2}.
\end{align}
Therefore, any VPD-invariant Lagrangian can be written as a function of three variables: $\gamma$, $a_1$, and $a_2$.

The Lagrangians with VPD, Diff, and Diff-Weyl symmetries are given, respectively, by
\begin{align}
    \mathcal{L}_{\VPD} &=  f\left( \sqrt{\gamma}, a_1, a_2 \right), \label{eq:lagrangian_vpd0} \\
    \mathcal{L}_{\Diff} &= \sqrt{\gamma} f\left( a_1, a_2 \right), \label{eq:lagrangian_diff} \\
    \mathcal{L}_{\DW} &= \sqrt{\gamma} \ f\left(\frac{a_2}{{a_1}^2} \right)  \label{eq:lagrangian_conf},
\end{align}
where the functions $f$ are generic functions of their arguments.
For example, the Schild action~\eqref{SS} is a special case of $\mathcal{L}_{\VPD}$, the Nambu-Goto action~\eqref{SNG} is a special case of $\mathcal{L}_{\Diff}$, and the Polyakov action~\eqref{SP} is a special case of $\mathcal{L}_{\DW}$.
The generalized Polyakov action defined by~\cite{Borissova:2024cpx}
\begin{align} \label{eq:Lagrangian_GP}
    \mathcal{L}_{GP} = F\left(\sqrt{h} \frac{1}{2} h^{ab} \gamma_{ab}\right)
    = F\left(\sqrt{h} \frac{a}{2}\right)
\end{align}
is a subclass of the Diff-Weyl-invariant actions~\eqref{eq:lagrangian_conf}.

In this and the next sections, we will prove the classical equivalence of these Lagrangians with the Nambu-Goto theory.
The strategy is to derive the equations of motion for the auxiliary field $h_{ab}$ in terms of $\gamma_{ab}$ and then substitute its solution back into the Lagrangian.

Using the definitions of $a_1$~\eqref{eq:a1} and $a_2$~\eqref{eq:a2}, we find
\begin{align}
    \frac{\partial a_1}{\partial h^{ij}} = \gamma_{ji}, \qquad \frac{\partial a_2}{\partial h^{ij}} = 2 \gamma_{ja} h^{ab} \gamma_{bi}.
\end{align}
For any Lagrangian $\mathcal{L}$ among~\eqref{eq:lagrangian_vpd0} --~\eqref{eq:lagrangian_conf}, the equation of motion for $h^{ij}$ is
\begin{equation}\label{eq:heom}
    \frac{\partial \mathcal{L}}{\partial h^{ij}} = \frac{\partial \mathcal{L}}{\partial a_1} \frac{\partial a_1}{\partial h^{ij}} + \frac{\partial \mathcal{L}}{\partial a_2} \frac{\partial a_2}{\partial h^{ij}}
    = \frac{\partial \mathcal{L}}{\partial a_1} \gamma_{ji} + 2 \frac{\partial \mathcal{L}}{\partial a_2} \gamma_{ja} h^{ab} \gamma_{bi} = 0.
\end{equation}
Multiplying both sides by $h_{kc}\gamma^{cj}$ yields 
\begin{align}\label{eq:heomn2}
    \frac{\partial \mathcal{L}}{\partial a_1} h_{ki} + 2 \frac{\partial \mathcal{L}}{\partial a_2} \gamma_{ki} = 0.
\end{align}
After dividing $\partial\mathcal{L}/\partial a_1$ on both sides, we find that the worldsheet metric $h_{ab}$ and the induced metric $\gamma_{ab}$ differ only by an overall factor $\phi$:
\begin{align}\label{eq:h_gamma}
    h_{ab} = \phi \gamma_{ab}
\end{align} 
with the proportionality factor given by
\begin{align}\label{eq:h_eom}
    \phi = \left.
    -2 \frac{\partial \mathcal{L}}{\partial a_2} \left( \frac{\partial \mathcal{L}}{\partial a_1}\right)^{-1} \right|_{a_1=2\phi^{-1}, a_2=2\phi^{-2}}.
\end{align}
Here, the values of $a_1$ and $a_2$ are determined in terms of $\phi$ from eqs. \eqref{eq:a1}, \eqref{eq:a2} and \eqref{eq:h_gamma} as
\begin{align}\label{eq:a1a2}
    a_1 = 2\phi^{-1}, \qquad a_2 = 2\phi^{-2}.
\end{align}
Thus, $\phi$ can be, in principle, solved from eq. \eqref{eq:h_eom}, and $h_{ab}$ is then determined by eq. \eqref{eq:h_gamma} in terms of $\gamma_{ab}$.
In the following, we proceed to discuss each case in~\eqref{eq:lagrangian_vpd0} --~\eqref{eq:lagrangian_conf} with the assumption that they can be solved with $\phi\neq 0$.
In Sec.~\ref{sec:ccexp}, we consider two exceptional cases where eq. \eqref{eq:h_eom} has no solution for $\phi$, as well as an example where a solution of $\phi$ exists.

\subsubsection{VPD-Invariant Actions}
\label{sec:VPDIA}

In the case of a VPD-invariant Lagrangian $\mathcal{L}_{\VPD}$~\eqref{eq:lagrangian_vpd0}, eq. \eqref{eq:h_eom} is equivalent to
\begin{align} \label{eq:phi_vpd}
    \phi = -2 \left. \frac{\partial f\left( \sqrt{\gamma}, a_1, a_2 \right)}{\partial a_2} \left( \frac{\partial f \left( \sqrt{\gamma}, a_1, a_2 \right)}{\partial a_1}\right)^{-1}\right\vert_{a_1=2 \phi^{-1},\ a_2=2 \phi^{-2}}
\end{align}
which may determine $\phi$ as a function of $\sqrt\gamma$. 
Substituting it back into eq. \eqref{eq:lagrangian_vpd0}, the Lagrangian becomes 
\begin{align} \label{eq:lagrangian_vpd}
    \mathcal{L}_{\VPD} = f\left(\sqrt{\gamma}, 2\phi^{-1}\left(\sqrt{\gamma}\right),  2\phi^{-2}\left(\sqrt{\gamma}\right)\right),
\end{align}
which is in the form of the generalized Schild Lagrangian $\mathcal{L}_{GS}$~\eqref{eq:generalized_schlid_action_opening_2}.
Therefore, generic VPD-invariant Lagrangians are equivalent to generalized Schild Lagrangians.

\subsubsection{Diffeomorphism-Invariant Actions}
\label{DIA}

As for the diffeomorphism-invariant Lagrangian  $\mathcal{L}_{\Diff}$~\eqref{eq:lagrangian_diff}, eq. \eqref{eq:h_eom} states that
\begin{align} \label{eq:phi_diffeo}
    \phi = -2 \left. \frac{\partial f\left( a_1, a_2 \right)}{\partial a_2} \left( \frac{\partial f \left( a_1, a_2 \right)}{\partial a_1}\right)^{-1}\right\vert_{a_1=2 \phi^{-1},\ a_2=2 \phi^{-2}}.
\end{align}
Unlike eq. \eqref{eq:phi_vpd} for the case of VPD symmetry, this equation fixes $\phi$ to a constant in the absence of $\sqrt{\gamma}$.
The resulting Lagrangian is therefore just the Nambu-Goto Lagrangian
\begin{align} \label{eq:lagrangian_diffeo}
    \mathcal{L}_{NG} = c \sqrt{\gamma}
\end{align}
with the constant $c \equiv f(a_1(\phi),a_2(\phi))$ and $\phi$ given by the constant solution of eq. \eqref{eq:phi_diffeo}.

\subsubsection{Diffeomorphism-Weyl-Invariant Actions}
\label{DWIA}

For Lagrangians $\mathcal{L}_{\DW}$~\eqref{eq:lagrangian_conf} invariant under Diff-Weyl, eq. \eqref{eq:h_eom} is trivially satisfied as follows
\begin{align}\label{eq:phi_conf1}
    \phi = \left.-2 \frac{\partial f\left( \frac{a_2}{{a_1}^2} \right)}{\partial a_2}\left( \frac{\partial f\left( \frac{a_2}{{a_1}^2} \right)}{\partial a_1} \right)^{-1} \right\vert_{a_1=2\phi^{-1},\ a_2=2\phi^{-2}}= \left.\frac{a_1}{a_2}\right\vert_{a_1=2\phi^{-1},\ a_2=2\phi^{-2}} = \phi.
\end{align}
That is, $\phi$ remains arbitrary.
This situation is the same as the Polyakov action, where the equation of motion for $h_{ab}$ only demands conformal equivalence between $h_{ab}$ and $\gamma_{ab}$ without fixing the proportionality factor $\phi$.

Although $\phi$ remains unfixed, the argument of the function $f$ in the Lagrangian is $\frac{a_2}{(a_1)^2} = \frac{1}{2}$ (after substituting $a_1=2\phi^{-1}$ and $a_2=2\phi^{-2}$).
Hence, the resulting Lagrangian is 
\begin{align}\label{eq:lagrangian_conf1}
    \mathcal{L}_{\DW} = f\left( \frac{1}{2} \right) \sqrt{\gamma},
\end{align}
which is again the Nambu-Goto Lagrangian.
Hence, both Diff-invariant Lagrangians $\mathcal{L}_{\Diff}$ and Diff-Weyl-invariant Lagrangians $\mathcal{L}_{\DW}$ are equivalent to the Nambu-Goto action.
In Sec.~\ref{sec:the Nambu-Goto_the_Schild_equiv}, we discuss the equivalence between eqs. \eqref{eq:generalized_schlid_action_opening_2} and~\eqref{eq:lagrangian_diffeo}.

\subsection{Examples}\label{sec:ccexp}

In this subsection, we explicitly solve the proportionality factor $\phi$ from eq. \eqref{eq:h_eom} for concrete examples of Lagrangians and illustrate two possibilities: (i) eq.~\eqref{eq:h_eom} admits a real, nonzero solution for $\phi$, and (ii) eq.~\eqref{eq:h_eom} is inconsistent and does not admit any solution for $\phi$.

\subsubsection{Polyakov $+$ Cosmological Constant term}

As the first example, we add a cosmological constant term $\lambda \sqrt{h}$ to the Polyakov Lagrangian:
\begin{align}
    \mathcal{L} = \frac{1}{2}\sqrt{h}h^{ab}\gamma_{ab} + \lambda \sqrt{h}.
\end{align}
Using eqs. \eqref{eq:a1} and~\eqref{gh-1}, we rewrite the Lagrangian in the form of eq. \eqref{eq:lagrangian_diff} with the function $f$ given by
\begin{align}
    f(a_1, a_2) = \frac{1}{\sqrt{2}}\frac{1}{\sqrt{{a_1}^2-a_2}}\left(a_1 + \frac{\lambda}{2}\right),
\end{align}
and~\eqref{eq:phi_diffeo} becomes
\begin{align}
    \phi = \frac{\phi^{-1}+2\lambda}{\phi^{-2}-2\lambda\phi^{-1}}.
\end{align}
This equation is equivalent to $\phi^{-1}-2\lambda = \phi^{-1} + 2\lambda$, which is self-inconsistent unless $\lambda = 0$.
Therefore, the cosmological constant deformation of the Polyakov action does not define a classical system with consistent solutions to its equations of motion.

\subsubsection{Polyakov $+$ $\text{(Polyakov)}^2$}

Next, we add a term quadratic in the Polyakov density:
\begin{align} \label{eq:Polyakov-squared}
    \mathcal{L} = \frac{1}{2}\sqrt{h}h^{ab}\gamma_{ab} + \frac{\mu}{4}\sqrt{h}\left( h^{ab}\gamma_{ab} \right)^2.
\end{align}
This Lagrangian can be rewritten in the form of eq. \eqref{eq:lagrangian_diff}, and the explicit form of the function $f$ is
\begin{align}
    f(a_1, a_2) = \frac{1}{\sqrt{2}}\frac{1}{\sqrt{{a_1}^2-a_2}} \left( a_1+\frac{\mu}{2} {a_1}^2 \right).
\end{align}
Equation~\eqref{eq:phi_diffeo} is then equivalent to
\begin{align}
    \phi^{-1}\mu = 0.
\end{align}
Assuming that $\mu \neq 0$, this implies that $\phi^{-1} = 0$.
Since this implies $h^{ab}\gamma_{ab}=0$, the Lagrangian becomes trivial for $\gamma_{ab}$.
At the same time, for finite $h_{ab}$, eq. \eqref{eq:h_gamma} implies $\gamma_{ab}=0$.
This suggests $X^\mu=\mathrm{const.}$ for all $\mu$, corresponding to a point-like string.
Hence, the Polyakov squared term by itself is also not considered a consistent modification to the Polyakov Lagrangian for a string.

\subsubsection{Polyakov $+$ Cosmological Constant term $+$ $\text{(Polyakov)}^2$}

Finally, we introduce the previous two deformations simultaneously:
\begin{align}
    \mathcal{L} = \frac{1}{2}\sqrt{h}h^{ab}\gamma_{ab} + \lambda \sqrt{h} + \frac{\mu}{4}\sqrt{h}\left( h^{ab}\gamma_{ab} \right)^2.
\end{align}
Rewriting the Lagrangian with $f$ in form of eq. \eqref{eq:lagrangian_diff}, we have
\begin{align}
    f(a_1, a_2) = \frac{1}{\sqrt{2}}\frac{1}{\sqrt{{a_1}^2-a_2}} \left( \frac{\lambda}{2} + a_1+\frac{\mu}{2} {a_1}^2 \right).
\end{align}
The constraint on $\phi$~\eqref{eq:phi_diffeo} becomes
\begin{align}
    -\frac{1}{2}\phi^2\lambda + \mu = 0.
\end{align}
Thus, for $\lambda\neq 0$ and $\mu/\lambda>0$ (so that $\phi^{2}>0$), eq. \eqref{eq:h_eom} admits real, nonzero solutions for $\phi$ ($\phi=\pm\sqrt{\frac{2\mu}{\lambda}}$), and the relation $h_{ab} = \phi\,\gamma_{ab}$ is well-defined. 
In this case, the deformation is consistent.

%%%%%%%%%%%%%%%%%%%%%%%%%%%%%%%%%%%%%%%%%%%%%%%%%%%%%%%%%%%%%%%%%%%%%%%%%%%%%%%%%%%%%%%%%%%%%%%%%%%%%%

\section{Equivalences of (Generalized) Schild and Nambu–Goto Actions}
\label{sec:the Nambu-Goto_the_Schild_equiv}

In Sec.~\ref{Riemann}, we classified general worldsheet actions built from the worldsheet metric $h_{ab}$ and the induced metric $\gamma_{ab}$, leading to the three symmetry classes \eqref{eq:lagrangian_vpd0} -- \eqref{eq:lagrangian_conf}, which are, respectively, invariant under the symmetries of VPD, Diff, and Diff-Weyl. 
In Secs.~\ref{DIA} and \ref{DWIA}, we have shown that diffeomorphism is sufficient to imply the equivalence with the Nambu-Goto action, regardless of whether there is the Weyl symmetry.
But actions with VPD symmetry are only shown to be equivalent to the generalized Schild action in Sec.~\ref{sec:VPDIA}.
In this section, we take the next step and show that the generalized Schild action is equivalent to the Nambu-Goto action, so that VPD-invariant actions are also equivalent to the Nambu-Goto action.
That is, VPD invariance already suffices to ensure equivalence with the Nambu-Goto action.

In Sec.~\ref{sec:equiv_the Schild_nambu_goto}, we establish the classical equivalence between the Schild action and the Nambu-Goto action for an arbitrary spacetime metric $g_{\mu\nu}$. 
We also determine the overall coefficient of the Schild action by demanding its stress-energy tensor to match the string tension specified by the Nambu-Goto action. 

In Sec.~\ref{sec:generalized_the Schild_ng_equiv}, we show that the generalized Schild action~\eqref{eq:generalized_schlid_action_opening_2} is equivalent to the Schild action, and thus it is also equivalent to the Nambu-Goto action.

\subsection{Equivalence of Schild and Nambu-Goto Actions}\label{sec:equiv_the Schild_nambu_goto}

We now prove the classical equivalence between the Schild and the Nambu-Goto theories with an arbitrary spacetime metric $g_{\mu\nu}$.
This is achieved by showing that their equations of motion are equivalent when we take a special gauge ($\sqrt{\gamma}= $ const.) for the Nambu-Goto theory.
Then, we match their stress-energy tensors to see how the tension is encoded in the Schild action.

\subsubsection*{Classical Equivalence}

We start from the Schild and the Nambu-Goto Lagrangians:
\begin{align}
    \mathcal{L}_{S}&=a\gamma, \\
    \mathcal{L}_{NG}&=T\sqrt\gamma, \label{LNG}
\end{align}
where $a$ is a real constant, $T$ is the string tension and $\gamma$ is (the absolute value of) the determinant of $\gamma_{ab}$ defined in~\eqref{eq:gamma}.
Their equations of motion with respect to $X^\mu$ are
\begin{align}
    0&=\epsilon^{ab}\partial_{a}\left(X^\nu\partial_b\sigma_{\mu\nu}\right),
    \label{eq:the Schild_eom}\\
    0&=\epsilon^{ab}\partial_{a}\left(\frac{1}{\sqrt\gamma}X^\nu\partial_b\sigma_{\mu\nu}\right),
    \label{eq:ng_eom}
\end{align}
respectively, where 
\begin{align}
\label{sigma-def}
\sigma^{\mu\nu}=\epsilon^{ab}\partial_a X^\mu\partial_b X^\nu.
\end{align} 
Here we have used 
\begin{align}\label{eq:gamma_sigma_square}
    \gamma=g_{\mu\rho}g_{\nu\sigma}\sigma^{\mu\nu}\sigma^{\rho\sigma}
\end{align}
and defined
\begin{align}
    \sigma_{\mu\nu}=g_{\mu\rho}g_{\nu\lambda}\sigma^{\rho\lambda}.
\end{align}
We shall denote $\gamma = \sigma^2$, as a shorthand for eq. \eqref{eq:gamma_sigma_square}.

These two equations of motion appear to be different since $\partial_a$ in eq. \eqref{eq:ng_eom} acts on $\frac{1}{\sqrt\gamma}$.
However, in the Nambu-Goto theory, since $\sqrt\gamma$ is a scalar density on the worldsheet, we can choose $\sqrt{\gamma}$ to be a worldsheet constant as a gauge-fixing condition for the diffeomorphism symmetry.
The residual gauge symmetry is VPD.

Hence, the equation of motion~\eqref{eq:ng_eom} of the Nambu-Goto theory together with the gauge-fixing condition $\sqrt{\gamma} = \mathrm{const.}$ implies the equation of motion~\eqref{eq:the Schild_eom} of the Schild theory.
To prove that the Schild theory is equivalent to the Nambu-Goto theory, we show that the gauge-fixing condition is implied by its equation of motion~\eqref{eq:the Schild_eom}.

We multiply $\partial_c X^\mu$ to eq. \eqref{eq:the Schild_eom} and get
\begin{equation}\label{eq:sigma_constant}
\begin{split}
        0
    &=\partial_cX^\mu\epsilon^{ab}\partial_a\left(X^\nu\partial_b\sigma_{\mu\nu}\right)\\
    &=\frac{1}{4}\epsilon^{ab}\epsilon_{ca}\sigma^{\mu\nu}\partial_b\sigma_{\mu\nu}\\
    &=-\frac{1}{8}\partial_c\gamma.
\end{split}
\end{equation}
In the second equality, we have used the fact that the indices $\mu,\nu$ in $\partial_cX^\mu\partial_aX^\nu$ is (totally) anti-symmetrized due to the anti-symmetrized tensor $\sigma_{\mu\nu}$, together with the identity
\begin{align}\label{eq:antisym_sum}
    \partial_cX^{\mu}\partial_aX^{\nu}-\partial_cX^{\nu}\partial_aX^{\mu}=\frac12  \epsilon_{ca}\sigma^{\mu\nu}.    
\end{align}

In the Schild theory, the relation $\sqrt\gamma = \mathrm{const.}$ is not a gauge-fixing condition, but rather it is dynamically fixed (by the equation of motion) everywhere on the same connected worldsheet once its initial condition is specified somewhere on the same worldsheet.
Thus, on any connected worldsheet, there is no physical difference between the Schild theory and the Nambu-Goto theory.

\subsubsection*{String Tension in Schild Action}

Next, we discuss how to define tension from the Schild action by matching the stress-energy tensors of the Schild action and the Nambu-Goto action.
To do so, we rewrite the actions so that they explicitly depend on $g_{\mu\nu}$ as follows.
\begin{align}
    S_{Schild} &= a\int d^2\xi \left( g^{\mu\rho}(x)g^{\nu\sigma}(x) \sigma_{\mu\nu}\sigma_{\rho\sigma} \right)\\
    S_{NG} &= T\int d^2\xi \ \sqrt{g^{\mu\rho}(x)g^{\nu\sigma}(x) \sigma_{\mu\nu}\sigma_{\rho\sigma}}.
\end{align}
Their stress-energy tensors are
\begin{equation}
\begin{split}
    T^{(Schild)}_{\alpha\beta}&=\frac{\delta S_{Schild}}{\delta g^{\alpha\beta}(y)}\\
    &=2a\int d^2\xi \left(\delta^{(D)}\left(x-y\right)g^{\nu\sigma}\sigma_{\alpha\nu}\sigma_{\beta\sigma}\right)\\
    T^{(NG)}_{\alpha\beta}&=\frac{\delta S_{NG}}{\delta g^{\alpha\beta}(y)}\\
    &=T\int d^2\xi \left(\frac{\delta^{(D)}\left(x-y\right)g^{\nu\sigma}\sigma_{\alpha\nu}\sigma_{\beta\sigma}}{\sqrt{\sigma^2}}\right)
\end{split}\label{eq:stress_energies}
\end{equation}
where $\delta^{(D)}(x-y)$ is the delta function in $D$-dimensional spacetime.
Equating the integrands in~\eqref{eq:stress_energies} yields
\begin{align}\label{eq:tension_relation}
    T=2a{\sqrt{\sigma^2}}.
\end{align}
Since both $a$ and $T$ are constants, this relation is only meaningful when $\sigma^2$ is constant, which is true on each connected worldsheet in the Schild theory as shown in eq. \eqref{eq:sigma_constant}. 
Then, the tension is given by the Schild parameter $a$ and the constant worldsheet measure $\sqrt{\sigma^2}$ through eq. \eqref{eq:tension_relation}.

\subsection{Equivalence between Generalized Schild and Nambu-Goto Actions}
\label{sec:generalized_the Schild_ng_equiv}

We have demonstrated in Sec.~\ref{sec:equiv_the Schild_nambu_goto} that the Schild action is equivalent to the Nambu-Goto action. 
Now we show that the generalized Schild action~\eqref{eq:generalized_schlid_action_opening_2} is equivalent to the Schild action, and thus it is equivalent to the Nambu-Goto action.

Starting from the generalized Schild action~\eqref{eq:generalized_schlid_action_opening_2}, we have the equation of motion for $X^\mu$
\begin{equation}\label{eq:general_the Schild_eom}
\begin{split}
    0
    &=\epsilon^{ab} \partial_bX^\nu \partial_a\left(F'\left(\sigma^2\right)\sigma_{\mu\nu} \right)\\
    &=\epsilon^{ab} \partial_bX^\nu \left(F''\left(\sigma^2\right)\left(\partial_a\sigma^2\right)\sigma_{\mu\nu}+F'\left(\sigma^2\right)\partial_a\sigma_{\mu\nu} \right).
\end{split}
\end{equation}
Multiplying it by $\partial_cX^\mu$ and using the identity~\eqref{eq:antisym_sum}, we obtain
\begin{align}\label{eq:sigma_constant_2}
    \partial_c\sigma^2\left(F'\left(\sigma^2\right)+2\sigma^2F''\left(\sigma^2\right)\right)=0.
\end{align}
For generic $F$, the term $F'\left(\sigma^2\right)+2\sigma^2F''\left(\sigma^2\right)$ does not vanish identically.
The equation above implies that 
\begin{align}\label{eq:s2=c}
\partial_c\sigma^2 = 0
\end{align}
identically on the worldsheet.
(See Appendix~\ref{sec:discrete_sigma} for a more detailed proof.)
Hence, $\sigma^2$ is a worldsheet constant.
For non-zero $F'(\sigma^2)$, the $F'(\sigma^2)$ in eq. \eqref{eq:general_the Schild_eom} can be moved outside the derivative.
It is then straightforward to see that the equation of motion~\eqref{eq:general_the Schild_eom} of the generalized Schild action matches the equation of motion \eqref{eq:the Schild_eom} for the Schild action.

Since we have shown that the Schild theory is equivalent to the Nambu-Goto theory in Sec.~\ref{sec:equiv_the Schild_nambu_goto}, the generalized Schild action~\eqref{eq:generalized_schlid_action_opening_2} is also equivalent to the Nambu-Goto action.\footnote{There are cases where the term $F'\left(\sigma^2\right)+2\sigma^2F''\left(\sigma^2\right)$ vanishes identically. However, this is only possible if $F\left(\sigma^2\right)=A\sqrt{\sigma^2}+B$, which is nothing but the Nambu-Goto action.}
Together with the result in Sec.~\ref{sec:VPDIA} that VPD-invariant actions are equivalent to the generalized Schild actions, we have now completed the proof that all the actions~\eqref{eq:lagrangian_vpd0} --~\eqref{eq:lagrangian_conf} are equivalent to the Nambu-Goto action~\eqref{LNG}.

%%%%%%%%%%%%%%%%%%%%%%%%%%%%%%%%%%%%%%%%%%%%%%%%%%%%%%%%%%%%%%%%%%%%%%%%%%%%%%%%%%%%%%%%%%%%%%%%%%%%
\section{The Nambu-Goto and Schild Actions on Areal-Metric Manifolds}
\label{sec:algebraic_areal_bosonic}

In previous sections, we established the classical equivalence among the Schild, the generalized Schild, and the Nambu-Goto actions on Riemannian target manifolds.
In this section, we extend this analysis to target 
manifolds equipped with an areal metric $G_{\mu\nu\rho\lambda}$.

Note that as we will show in Sec.~\ref{sec:derivation_of_general_lagrangians}, any nontrivial dependence on the auxiliary areal worldsheet degrees of freedom is forbidden by Weyl symmetry.
The only possible Polyakov-like action (those with both diffeomorphisms and Weyl symmetry) is the one discussed in Sec.~\ref{sec:anomalous_polyakov}.
Hence, we only include the Nambu-Goto and Schild actions in this section.

An areal metric, as a generalization of the ordinary metric, is defined as eq. \eqref{eq:areal_metric_def}:
\begin{align}
    dA^2 = G_{\mu\nu\rho\lambda} da^{\mu\nu} da^{\rho\lambda}
    \label{dA2}
\end{align}
where $dA$ is the infinitesimal area element, $G_{\mu\nu\rho\sigma}$ is the areal metric, and $da^{\mu\nu}$ is a 2-form on the manifold. 

From its definition, we see that $G_{\mu\nu\rho\sigma}$ must have the following symmetries:
\begin{align} \label{eq:areal-metric-symm}
    G_{\mu\nu\rho\lambda} = -G_{\nu\mu\rho\lambda} = -G_{\mu\nu\lambda\rho} = G_{\rho\lambda\mu\nu}.
\end{align}
In a $D$-dimensional spacetime, it has $\frac{ D\left(D-1\right)\left(D^2-D+2\right)}{8}$ degrees of freedom.
By comparison, an ordinary metric has $\frac{D(D+1)}{2}$ degrees of freedom.
An areal metric contains more degrees of freedom for $D > 3$.
It is also conventional to impose the cyclicity condition
\begin{align}
     G_{\mu\nu\rho\lambda} + G_{\mu\rho\lambda\nu} + G_{\mu\lambda\nu\rho}  = 0
\end{align}
so that it is an irreducible representation of the local frame group $SL(D, \mathbb{R})$~\cite{Schuller:2005yt}.\footnote{That is, the areal metric $G_{\mu\nu\lambda\rho}$ cannot be used to define a nontrivial product between 1-forms and 3-forms, or a function on the space of 4-forms.}
This condition reduces the number of free parameters in an areal metric to $\frac{D^2 (D^2 - 1)}{12}$ in $D$ dimensions.
It still has more degrees of freedom than an ordinary metric for $D > 3$.

Since this class of objects has more degrees of freedom than ordinary metrics when $D>3$, an areal-metric in target space ($D\geq 4$) in general cannot be captured by a Riemannian metric.
In the following, we demonstrate how the Nambu-Goto and Schild  actions can be generalized to those defined for areal-metric backgrounds; we also discuss a Polyakov-like auxiliary-field formulation.
We will also prove their classical equivalence.

In Sec.~\ref{sec:nambu-area}, we introduce the areal-metric analogue of the Nambu-Goto action and derive its equation of motion.
On areal-metric manifolds, we consider  the Schild action with VPD symmetry in Sec.~\ref{sec:gng=gs} and a diffeomorphism-invariant action with an auxiliary worldsheet field as the counterpart of the Polyakov action in Sec.~\ref{sec:areal_nambu_polyakov}.
We will prove their equivalence to the Nambu-Goto action on areal-metric manifolds.

In Sec.~\ref{sec:derivation_of_general_lagrangians}, we will discuss more general actions on areal metric manifolds.

\subsection{Nambu-Goto action on Areal-Metric Manifold}\label{sec:nambu-area}

Due to the geometrical meaning of the Nambu-Goto action, it is straightforward to define the analogue of the Nambu-Goto action for an areal-metric background to be given by the worldsheet area as
\begin{align}\label{eq:action_GNG}
    S^A_{NG} = \int d^2\xi \sqrt{
    \frac{1}{4} G_{\mu\nu\rho\lambda} \sigma^{\mu\nu}\sigma^{\rho\lambda}},
\end{align}
where $G_{\mu\nu\rho\lambda}$ is the areal metric, $\sigma^{\mu\nu} = \epsilon^{ab} \partial_a X^\mu \partial_b X^\nu$~\eqref{sigma-def}, and the Greek (Roman) letters are spacetime (worldsheet) indices.
This action reduces to the usual Nambu-Goto action on Riemannian spacetime when the areal metric agrees with
\begin{align}\label{eq:area_to_line}
G_{\mu\nu\rho\lambda} = g_{\mu\rho}g_{\nu\lambda} - g_{\mu\lambda}g_{\nu\rho}
\end{align}
for a given Riemannian metric $g_{\mu\nu}$.
We will refer to~\eqref{eq:action_GNG} as the Nambu-Goto action for the areal metric $G$.

\subsection{The Schild Action on Areal-Metric Manifolds}\label{sec:gng=gs}

Since the Schild Lagrangian~\eqref{SS} can be obtained by squaring the Nambu-Goto Lagrangian~\eqref{SNG}, we define the areal-metric generalization of the Schild action as
\begin{align}\label{eq:action_GS}
    S^A_{S} = \int d^2\xi \,
    \frac{1}{4} G_{\mu\nu\rho\lambda} \sigma^{\mu\nu}\sigma^{\rho\lambda} .
\end{align}
It respects the VPD symmetry.
We now show that $S^A_{S}$ is classically equivalent to the areal-metric Nambu-Goto action $S^A_{NG}$ in \eqref{eq:action_GNG}, for an arbitrary background areal metric $G_{\mu\nu\rho\lambda}$.
The proof follows essentially the same line as that for Riemannian manifolds.

The equation of motion for $X^\alpha$ obtained from $S^A_{S}$~\eqref{eq:action_GS} is
\begin{align}\label{eq:eom_GS}
    0=\frac{1}{4}\partial_\alpha G_{\mu\nu\rho\lambda} \sigma^{\mu\nu} \sigma^{\rho\lambda} - \partial_a \left( G_{\alpha\nu\rho\lambda} \epsilon^{ab} \partial_b X^\nu \sigma^{\rho\lambda} \right).
\end{align}
After multiplying $\partial_c X^\alpha$ on both sides of eq. \eqref{eq:eom_GS} and using eq. \eqref{eq:antisym_sum},
we obtain
\begin{align}\label{eq:areal_G_sigma_sigma}
    0 = \partial_c\left( \sigma_A^2 \right),
\end{align}
where
\begin{align}
    \sigma_A^2 \equiv \frac{1}{4} G_{\mu\nu\rho\lambda} \sigma^{\mu\nu} \sigma^{\rho\lambda}.
    \label{sigma2-def}
\end{align}
we get from \eqref{eq:areal_G_sigma_sigma} that
\begin{align}
    \sigma_A^2 = \mbox{constant}.
    \label{s2=const}
\end{align}
This is in complete analogy with the original Schild action~\eqref{SS} whose equation of motion implies that $\sigma^2=g_{\mu\rho}g_{\nu\lambda} \sigma^{\mu\nu} \sigma^{\rho\lambda}$ is a worldsheet constant.

For the case of $S^A_{NG}$~\eqref{eq:action_GNG}, the equation of motion for $X^\alpha$ is given by
\begin{align} \label{eq:eom_gng}
    0 = \frac{1}{2\sqrt{\sigma_A^2}}\partial_\alpha G_{\mu\nu\rho\lambda} \sigma^{\mu\nu} \sigma^{\rho\lambda} - \partial_a \left( \frac{1}{\sqrt{\sigma_A^2}} G_{\alpha\nu\rho\lambda} \epsilon^{ab} \partial_b X^\nu \sigma^{\rho\lambda} \right).
\end{align}
Since $\sqrt{\sigma_A^2}$ is a scalar density on the worldsheet, we can fix it to a constant via part of the diffeomorphisms, reducing the symmetry to VPD.
This allows us to factor out $1/\sqrt{\sigma_A^2}$ from $\partial_a$ in the second term.
After multiplying eq. \eqref{eq:eom_gng} by $\sqrt{\sigma_A^2}$, the equation becomes
\begin{align}
    0 = \frac{1}{2}\partial_\alpha G_{\mu\nu\rho\lambda} \sigma^{\mu\nu} \sigma^{\rho\lambda} - \partial_a \left( G_{\alpha\nu\rho\lambda} \epsilon^{ab} \partial_b X^\nu \sigma^{\rho\lambda} \right)
\end{align}
which is the same as eq. \eqref{eq:eom_GS}, the equation of motion for $X^\mu$ from the areal Schild action $S^A_{S}$.
Thus, we have proven the classical equivalence for any areal metric between the Schild action and the Nambu-Goto action.

\subsection{Action with Auxiliary Field on Areal-Metric Manifolds}
\label{sec:areal_nambu_polyakov}

We intend to construct an action analogous to the Polyakov action, which is an action involving not only the induced metric $\gamma_{ab}$ but also an auxiliary worldsheet metric $h_{ab}$.
While the spacetime metric is now promoted from $g_{\mu\nu}$ to $G_{\mu\nu\lambda\rho}$, the auxiliary worldsheet metric $h_{ab}$ should also be promoted to an areal metric $\mathcal{H}_{abcd}$.
However, in 2 dimensions, $\mathcal{H}_{abcd}$ has only one independent component $\mathcal{H}_{0101}$, which we simply denote as $\mathcal{H}$.
Using this auxiliary field $\mathcal{H}$, we can remove the square root in~\eqref{eq:action_GNG} with the full diffeomorphism preserved.
It turns out that it is impossible to enhance the diffeomorphism to have Weyl symmetry or scale symmetry.

Analogous to the einbein action for particles
\begin{align}
    S_{e}=\frac{1}{2}\int d\tau\left(e^{-1}\dot{X}^2+em^2\right),
\end{align}
we consider
\begin{align}
\label{eq:action_GP0}
S^A_{e} = \frac{1}{2} \int d^2\xi \, \sqrt{\mathcal{H}} \left[
\mathcal{H}^{-1} \left(\frac{1}{4} G_{\mu\nu\rho\lambda} \sigma^{\mu\nu} \sigma^{\rho\lambda}\right) + T^2\right].
\end{align}
Here, we regard $\sqrt{\mathcal{H}}$ as the analogue of the einbein $e$, $\frac{1}{4} G_{\mu\nu\rho\lambda} \sigma^{\mu\nu} \sigma^{\rho\lambda}$ as that of $\dot{X}^2$ and the tension $T$ as that of the mass $m$.
Since $\sqrt{\mathcal{H}}$ transforms as a scalar density under worldsheet diffeomorphism, this action is invariant under worldsheet diffeomorphism.

To demonstrate its equivalence to the Nambu-Goto action~\eqref{eq:action_GNG} for the areal metric $G_{\mu\nu\rho\lambda}$, we integrate out the worldsheet areal metric $\mathcal{H}$.
By varying the action~\eqref{eq:action_GP0} with respect to $\sqrt{\mathcal{H}}$, we get the equation of motion for $\mathcal{H}$:
\begin{align}
    0=\frac{-1}{{\left(\sqrt{\mathcal{H}}\right)}^2}\left(\frac{1}{4}G_{\mu\nu\rho\lambda}\sigma^{\mu\nu}\sigma^{\rho\lambda}\right)+T^2
\end{align}
which is solved by
\begin{align}
    \sqrt{\mathcal{H}}=\frac{1}{T}\sqrt{\frac{1}{4}G_{\mu\nu\rho\lambda}\sigma^{\mu\nu}\sigma^{\rho\lambda}}.
\end{align}
Substituting this back into eq. \eqref{eq:action_GP0} gives
\begin{align}
    S^A_{e} = T\int d^2\xi \, \sqrt{\frac{1}{4} G_{\mu\nu\rho\lambda} \sigma^{\mu\nu} \sigma^{\rho\lambda}},
\end{align}
which is exactly the Nambu-Goto action on areal-metric manifolds~\eqref{eq:action_GNG} with tension $T$.

%%%%%%%%%%%%%%%%%%%%%%%%%%%%%%%%%%%%%%%%%%%%%%%%%%%%%%%%%%%%%%%%%%%%%%%%%%%%%%%%%%%%%%%%%%%%%%%%%%%%
\section{General Worldsheet Actions on Areal-Metric Manifolds} \label{sec:derivation_of_general_lagrangians}

In this section, we extend the discussion in Secs.~\ref{Riemann} and~\ref{sec:the Nambu-Goto_the_Schild_equiv} from a Riemannian metric background to an areal-metric background.
We show that, in the areal-metric case, the most general VPD-invariant actions and diffeomorphism-invariant actions are both classically equivalent to the areal Nambu-Goto action~\eqref{eq:action_GNG}.

Let $G_{\mu\nu\rho\lambda}$ be the areal metric in spacetime, $\xi^a$ the worldsheet coordinates, and $X^\mu\left(\xi\right)$ the embedding functions. 
The induced areal metric on the worldsheet is
\begin{align}
    \gamma_{abcd} \equiv G_{\mu\nu\rho\lambda} \partial_a X^{\mu} \partial_b X^{\nu} \partial_c X^{\rho} \partial_d X^{\lambda}
    = \epsilon_{ab} \epsilon_{cd} \sigma_A^2,
\end{align}
where $\sigma^{\mu\nu}$ and $\sigma_A^2$ are respectively defined in eqs.~\eqref{sigma-def} and \eqref{sigma2-def}.
For convenience, we repeat the definitions here:
\begin{align*}
    \sigma_A^2 \equiv \frac{1}{4} G_{\mu\nu\rho\lambda} \sigma^{\mu\nu} \sigma^{\rho\lambda},
    \ \ \ \ \ 
    \sigma^{\mu\nu} \equiv \epsilon^{ab} \partial_a X^{\mu} \partial_b X^{\nu}.
\end{align*}
Since $\sigma^{\mu\nu}$ contains a single $\epsilon^{ab}$ and the areal metric $G_{\mu\nu\rho\lambda}$ is a background tensor, the quantity $\sigma_A^2$ transforms under worldsheet diffeomorphism as a scalar density of the same weight (weight $2$) as the auxiliary field $\mathcal{H}$ defined in Sec.~\ref{sec:areal_nambu_polyakov}.

Therefore, the most general VPD-invariant actions and diffeomorphism-invariant actions built from an areal metric are, respectively,
\begin{align}
    S^{A}_{\text{VPD}} &= \int d^2 \xi \,  F(\mathcal{H}, \sigma_A^2)\label{eq:action_s_vpd},
    \\
    S^{A}_{\text{Diff}} &= \int d^2 \xi \, \sqrt{\mathcal{H}} F\left(\frac{\sigma_A^2}{\mathcal{H}}\right), \label{eq:action_s_diff}
\end{align}
where $F$ is an arbitrary function.

In analogy to the ordinary case of Weyl transformation $h_{ab}\rightarrow e^\varphi h_{ab}$, we define Weyl scaling for the worldsheet areal metric $\mathcal{H}_{abcd}$ as
\begin{align}
    \mathcal{H}_{abcd}\rightarrow e^{2\varphi} \mathcal{H}_{abcd},
\end{align}
which is equivalent to $\mathcal{H} \rightarrow e^{2\varphi} \mathcal{H}$.
The factor $2$ in the exponent is chosen so that this transformation is consistent with the Weyl transformation when the worldsheet areal metric is defined by a worldsheet metric $h_{ab}$ as
\begin{align}
    \mathcal{H}_{abcd}=h_{ac}h_{bd}-h_{ad}h_{bc}.
\end{align}
Under this Weyl scaling, $\sigma_A^2$ is invariant. 
Weyl symmetry implies that the Lagrangian is independent of $\mathcal{H}$ while the Weyl symmetry acts trivially on $\sigma_A^2$ in the action.

\subsection{Diffeomorphism-Invariant Areal Actions}

We first consider the diffeomorphism-invariant action~\eqref{eq:action_s_diff}.
To eliminate $\mathcal{H}$ from eq.~\eqref{eq:action_s_diff}, we use the equation of motion for $\mathcal{H}$,
\begin{align}\label{eq:H_F}
    \frac{1}{2\sqrt{\mathcal{H}}}\left(F(y)-2y F'(y)\right) = 0,
    \qquad y \equiv \frac{\sigma_A^2}{\mathcal{H}}
\end{align}
where $F'$ denotes the derivative of $F$ with respect to its argument.

In the following, we demonstrate how the $\mathcal{H}$-dependence of eq.~\eqref{eq:action_s_diff} can be removed using eq.~\eqref{eq:H_F}. 
First, we multiply $\partial_a y$ on both sides and get
\begin{align}\label{eq:F_diff_areal}
    \frac{F^3\left(y\right)}{2\sqrt{\mathcal{H}}}\partial_a\left(\frac{y}{F^2\left(y\right)}\right) = 0.
\end{align}
For $F\left(y\right)\neq 0$, the term $\frac{y}{F^2\left(y\right)}$ must be a worldsheet constant,\footnote{There can be some $\frac{\sigma_A^2}{\mathcal{H}}=y_0$ such that $F(y_0)=0$. 
This implies $F'(y_0)=0$ from \eqref{eq:H_F}. 
Then, they ensure that the equation of motion for $X^\mu$ is satisfied automatically. 
In this case, the only equation we have is $\frac{\sigma_A^2}{\mathcal{H}}=y_0$. 
Then, any $X^\mu$ is allowed, as long as we set $\mathcal{H}$ to be $\frac{\sigma_A^2}{y_0}$, which is not a physically interesting solution.} say $C^{-2}$
Since the term being worldsheet constant $C^{-2}$ means that
\begin{align}\label{eq:h_sub_sigma}
    \sqrt{\mathcal{H}}F\left(\frac{\sigma_A^2}{\mathcal{H}}\right)= C \sqrt{\sigma_A^2},
\end{align}
we substitute \eqref{eq:h_sub_sigma} into 
\eqref{eq:action_s_diff}, and get
\begin{align}\label{eq:sadiff}
    S_{A,\text{Diff}}=C\int d^2\xi \sqrt{\sigma_A^2},
\end{align}
which is exactly the Nambu-Goto action on areal-metric manifolds~\eqref{eq:action_GNG} with tension $C$.

There are two classes of solutions, depending on whether $y$ is a constant on the worldsheet.
If $y$ is not a constant,
we can solve eq.~\eqref{eq:H_F} as a differential equation of $F$ as a function of $y$ by
\begin{align}\label{eq:FH}
   F(y) = C \sqrt{y} = C \sqrt{\frac{\sigma_A^2}{\mathcal{H}}},
\end{align}
for an arbitrary constant $C$.
The action $S^A_{\text{Diff}}$~\eqref{eq:action_s_diff} is thus 
\begin{align}
S^{A}_{\text{Diff}} &= C \int d^2 \xi \, \sqrt{\sigma_A^2},
\end{align}
which coincides with the areal Nambu-Goto action~\eqref{eq:action_GNG} with tension $C$.

The other class of solutions to eq.~\eqref{eq:H_F} corresponds to its algebraic solutions $\{y_i\}$.
Then, we can solve eq.~\eqref{eq:H_F} by setting ${\cal H} = \sigma_A^2/y_i$.
The resulting action is
\begin{align}
    S^A_{\text{Diff}}= \int d^2\xi  \, \sqrt{\mathcal{H}} F(y_i) = \frac{F(y_i)}{\sqrt{y_i}} \int d^2\xi  \, \sqrt{\sigma_A^2}.
\end{align}
This is the areal Nambu-Goto action~\eqref{eq:action_GNG} with tension $F(y_i)/\sqrt{y_i}$.
Hence, except singular cases when $F(y_i) = 0$ or $y_i = 0$, any diffeomorphism-invariant areal action of the form \eqref{eq:action_s_diff} is classically equivalent to the areal Nambu–Goto theory.

\subsection{VPD-Invariant Areal Actions}

We next study the VPD-invariant action $S^{A}_{\text{VPD}}$~\eqref{eq:action_s_vpd}.
Varying it with respect to $\mathcal{H}$ yields
\begin{align}
    \frac{\partial F(\mathcal{H}, \sigma_A^2)}{\partial \mathcal{H}} = 0.
\end{align}
This implies that $F$ must in fact be independent of $\mathcal{H}$, and thus
\begin{align}\label{eq:FHSIGMAG}
    F(\mathcal{H}, \sigma_A^2) = f(\sigma_A^2)
\end{align}
for some arbitrary function $f$.
The action \eqref{eq:action_s_vpd} therefore reduces to
\begin{align} \label{eq:AM-the Schild}
    S^A_{GS} \equiv \int d^2\xi \, f(\sigma_A^2)
\end{align}
which we call the generalized areal Schild action.

To show that $S^A_{GS}$~\eqref{eq:AM-the Schild} is classically equivalent to the areal Nambu-Goto action~\eqref{eq:action_GNG}, we vary $S^A_{GS}$ with respect to $X^\alpha$.
This gives 
\begin{align} \label{eq:SAS-eom}
    0 = \frac{1}{4} f'(\sigma_A^2) \left(\partial_{\alpha} G_{\mu\nu\rho\lambda}\right) \sigma^{\mu\nu} \sigma^{\rho\lambda} - \partial_a \left(f'(\sigma_A^2) G_{\alpha\nu\rho\lambda} \epsilon^{ab} \partial_b X^{\nu} \sigma^{\rho\lambda}\right)
\end{align}
where $f'$ denotes the derivative of $f$ with respect to its argument $\sigma_A^2$.

Multiplying eq.~\eqref{eq:SAS-eom} by $\partial_c X^{\alpha}$ and using the identity \eqref{eq:antisym_sum}, we obtain
\begin{align}
    0 &= \frac{1}{4} f'(\sigma_A^2) \left(\partial_c G_{\mu\nu\rho\lambda}\right) \sigma^{\mu\nu} \sigma^{\rho\lambda} - \epsilon^{ab} \partial_c X^{\alpha} \partial_b X^{\nu} \partial_a \left(f'(\sigma_A^2) G_{\alpha\nu\rho\lambda} \sigma^{\rho\lambda}\right) \\
    &= \frac{1}{4} f'(\sigma_A^2) \left(\partial_c G_{\mu\nu\rho\lambda}\right) \sigma^{\mu\nu} \sigma^{\rho\lambda} - \epsilon^{ab} \frac{1}{2} \epsilon_{cb} \sigma^{\alpha\nu} \partial_a \left(f'(\sigma_A^2) G_{\alpha\nu\rho\lambda} \sigma^{\rho\lambda}\right).
\end{align}
This can be rewritten as 
\begin{align}
    0&= \frac{1}{4} f'(\sigma_A^2) \left(\partial_c G_{\mu\nu\rho\lambda}\right) \sigma^{\mu\nu} \sigma^{\rho\lambda} - \frac{1}{2} \sigma^{\alpha\nu} \partial_c \left(f'(\sigma_A^2) G_{\alpha\nu\rho\lambda} \sigma^{\rho\lambda}\right) \\
    &= - \frac{1}{4f'(\sigma_A^2)} \partial_c \left({f'}^2(\sigma_A^2) \sigma_A^2\right),
\end{align}
or, equivalently,
\begin{align}\label{eq:eom-1}
    0&= \left(\partial_c \sigma_A^2\right)  \partial_{\sigma_A^2}\left({f'}^2(\sigma_A^2) \sigma_A^2\right).
\end{align}

For the case of the areal Nambu-Goto action with $f(\sigma_A^2) = T \sqrt{\sigma_A^2}$, the equation above is satisfied identically.
For other choices of $f$, the factor $\partial_{\sigma_A^2}\left(f'^2\left(\sigma_A^2\right)\sigma_A^2\right)$ does not vanish identically, so eq.~\eqref{eq:eom-1} implies
\begin{align}\label{eq:constant_gamma_4}
    \partial_c\left(\sigma_A^2\right)=0.
\end{align}
Thus, $\sigma_A^2$ must be constant on the worldsheet (see Appendix~\ref{sec:discrete_sigma} for a detailed discussion).
Consequently, $f'(\sigma_A^2)$ has to be a worldsheet constant as well.
We can therefore take $f'(\sigma_A^2)$ outside the derivative in the second term of eq.~\eqref{eq:SAS-eom}.
The resulting equation of motion has the same form as that derived from $S^A_{S}$ \eqref{eq:action_GS}, which was shown in Sec.~\ref{sec:gng=gs} to be equivalent to the areal Nambu-Goto theory.

We conclude that all areal actions are classically equivalent to the areal Nambu-Goto action.

%%%%%%%%%%%%%%%%%%%%%%%%%%%%%%%%%%%%%%%%%%%%%%%%%%%%%%%%%%%%%%%%%%%%%%%%%%%%%%%%%%%%%%%%%%%%%%%%%%%%

\section{Areal-Metric Deformation of Background Metric}
\label{sec:anomalous_polyakov}

The Polyakov action on a Riemannian background describes a critical string, but it is not clear whether this statement holds on an areal-metric manifold.

To address this question, we consider a small areal-metric deformation of an ordinary metric $g_{\mu\nu}$, parametrized by
\begin{align}
\label{eq:G=g+a}
G_{\mu\nu\rho\lambda} = g_{\mu\rho}g_{\nu\lambda} - g_{\mu\lambda}g_{\nu\rho} + a_{\mu\nu\rho\lambda},
\end{align}
where $a_{\mu\nu\rho\lambda}$ is treated as a perturbation.
The following action, proposed in ref.~\cite{Borissova:2024cpx}, generalizes the Polyakov action: 
\begin{equation}\label{eq:action_GP}
    S_{GP} 
    =\int d^2\xi \ \mathcal{L}_{GP},\ \ \ \ \ \ 
    \mathcal{L}_{GP}\equiv \sqrt{\left( \frac{\sqrt{-h}}{2} h^{ab} \gamma_{ab} \right)^2 + \frac{1}{4} a_{\mu\nu\rho\lambda} \sigma^{\mu\nu} \sigma^{\rho\lambda}},
\end{equation}
where $h_{ab}$ and $\gamma_{ab}$ are the worldsheet and induced metrics defined in Sec.~\ref{s:intro}.
The action~\eqref{eq:action_GP} is invariant under two-dimensional diffeomorphisms and Weyl scaling $h_{ab} \rightarrow e^{\varphi} h_{ab}$.

\subsection{Classical Equivalence to Areal Nambu-Goto}

We first show that the action \eqref{eq:action_GP} is classically equivalent to the areal Nambu-Goto action~\eqref{eq:action_GNG}.
Varying $S_{GP}$ with respect to $h_{ab}$ gives
\begin{align}
    0 = \frac{1}{\mathcal{L}_{GP}}\frac{1}{2}\left( -\frac{1}{2}\sqrt{-h} h^{cd} \gamma_{cd} h_{ab} + \sqrt{-h} \gamma_{ab} \right).
\end{align}
This is the same relation as in the ordinary Polyakov theory, and implies
\begin{align}
    \frac{h_{ab}}{\sqrt{-h}}=\frac{\gamma_{ab}}{\sqrt{\gamma}}.
\end{align}
Substituting this solution back to eq.~\eqref{eq:action_GP} reproduces the areal Nambu-Goto action~\eqref{eq:action_GNG}.
Thus, at the classical level, the generalized Polyakov action \eqref{eq:action_GP} and the areal Nambu-Goto theory are equivalent.

\subsection{Quantum Consistency and Primary Operator Condition}

We now investigate the quantum consistency.
In the conformal gauge $h_{ab} = \eta_{ab}$, the leading order of eq.~\eqref{eq:action_GP} reduces to the ordinary Polyakov action~\eqref{SP} of the bosonic string,
\begin{align} \label{eq:CFT0}
    S = \frac{1}{2} \int d^2\xi \, \eta^{ab} g_{\mu\nu} \partial_a X^{\mu} \partial_b X^{\nu},
\end{align}
which defines a conformal field theory.
A deformation of the action preserves the conformal symmetry at the quantum level if the corresponding operator is a primary field of weight $\left(1,1\right)$.

As pointed out in ref.~\cite{Borissova:2024cpx}, when $a_{\mu\nu\rho\lambda}$ is small, the action \eqref{eq:action_GP} can be treated as a perturbation of eq.~\eqref{eq:CFT0} by a fractional term
\begin{align} \label{operator-a}
    \frac{a_{\mu\nu\rho\lambda} \epsilon^{ab} \epsilon^{cd} \partial_a X^{\mu} \partial_b X^{\nu} \partial_c X^{\rho} \partial_d X^{\lambda}}{g_{\alpha\beta} \partial_a X^{\alpha} \partial_b X^{\beta}}
\end{align}
at the leading order of $a_{\mu\nu\rho\lambda}$.

We now test whether this operator is a $(1,1)$-primary operator on the flat background 
\begin{align}
    g_{\alpha\beta} = \eta_{\alpha\beta}.
\end{align}
Using complex coordinates $\left(z,\bar{z}\right)$ on the worldsheet, the operator~\eqref{operator-a} becomes
\begin{align} \label{eq:area_metric_operator}
    \mathcal{O} \equiv \, \normord{\frac{1}{\eta_{\alpha\beta} \partial_z X^\alpha \partial_{\bar{z}} X^\beta} a_{\mu\nu\rho\lambda} \partial_z X^\mu \partial_{\bar{z}} X^\nu \partial_z X^\rho \partial_{\bar{z}} X^\lambda}.
\end{align}

In general, consider a normal-ordered operator of the form
\begin{align}
    \mathcal{O} = \ \normord{F\left( \partial_z X^\mu \right)}    
\end{align}
and the holomorphic energy momentum tensor $T(z)$ of the free boson theory.
Their OPE takes the form 
\begin{align*}
    T(z) \mathcal{O}(w) = \frac{-1}{2} \frac{1}{\left( z-w \right)^4} \normord{\frac{\partial^2 F}{\partial\left(\partial_w X^\mu\right) \partial\left(\partial_w X_\mu\right)} }- \frac{-1}{\left( z-w \right)^2} \normord{\partial_w X^{\mu} \frac{\partial F}{\partial\left(\partial_w X^\mu\right)}} + \dots
\end{align*}
where $F=F\left( \partial_w X^\lambda \right)$, and the ellipsis are terms that are regular when $z\rightarrow w$.
For $\mathcal{O}$ to be a primary with conformal dimension $\Delta$, we must have
\begin{align} 
    \eta_{\mu\nu} \frac{\partial^2 F}{\partial\left(\partial_z X_\mu\right) \partial\left(\partial_z X_\nu\right)} = 0 \label{eq:primary_constraint_1},
    \\
    \partial_z X^{\mu} \frac{\partial F}{\partial\left(\partial_z X^\mu\right)} = \Delta F. \label{eq:primary_constraint_2}
\end{align}

In our case, we write these conditions in terms of 
\begin{align}
    u^\mu=\partial_zX^\mu, \ \ \ v^\mu=\partial_{\bar{z}} X^\mu,
\end{align}
and define 
\begin{align}\label{FKuu/uv}
    F(u,v) \equiv \frac{K_{\mu\nu} u^\mu u^\nu}{u\cdot v},
\end{align}
where
\begin{align}
    K_{\mu\nu} \equiv a_{\mu\rho\nu\lambda} v^\rho v^\lambda, \ \ u\cdot v \equiv \eta_{\mu\nu} u^\mu v^\nu.
\end{align}

The condition \eqref{eq:primary_constraint_2} is easy to check by counting the power of $u^\mu$ in eq.~\eqref{FKuu/uv}.
That is, there is one $u^\mu$ in the denominator and two in the numerator, so $\Delta = 1$.

The nontrivial constraint is condition \eqref{eq:primary_constraint_1}, which is equivalent to
\begin{align}\label{eq:primary_constraint_11}
    0 = \eta^{\mu\nu} \frac{\partial^2 F}{\partial u^{\mu} \partial u^{\nu}}.
\end{align}
A straightforward calculation gives
\begin{align}\label{eq:eta_f_uu}
    \eta^{\mu\nu} \frac{\partial^2 F(u,v)}{\partial u^{\mu} \partial u^{\nu}}
    &= \frac{\eta^{\mu\nu} K_{\mu\nu}}{u \cdot v} - 2 \frac{K_{\mu\nu} u^\mu v^\nu}{\left(u \cdot v\right)^2} + \frac{K_{\mu\nu} u^\mu u^\nu \eta_{\alpha\beta} v^\alpha v^\beta}{\left( u \cdot v \right)^3}.
\end{align}
Since $a_{\mu\nu\rho\lambda}=-a_{\nu\mu\rho\lambda}$, we have
\begin{align}
    K_{\mu\nu}v^\mu=a_{\mu\rho\nu\lambda}v^\mu v^\rho v^\lambda=0,
\end{align}
and eq.~\eqref{eq:eta_f_uu} simplifies to
\begin{align}
    \eta^{\mu\nu} \frac{\partial^2 F(u,v)}{\partial u^{\mu} \partial u^{\nu}}
    &= \frac{\eta^{\mu\nu} K_{\mu\nu}}{u \cdot v}  + \frac{K_{\mu\nu} u^\mu u^\nu \eta_{\alpha\beta} v^\alpha v^\beta}{\left( u \cdot v \right)^3}.
\end{align}
Therefore, the condition~\eqref{eq:primary_constraint_11} is equivalent to
\begin{align}\label{eq:G_primary_constraint}
    a_{\mu\rho\nu\lambda} \eta^{\mu\nu} v^\rho v^\lambda \left( u \cdot v \right)^2 + a_{\mu\rho\nu\lambda} u^\mu u^\nu v^\rho v^\lambda v^2 = 0.
\end{align}
We now show that the only solution to eq.~\eqref{eq:G_primary_constraint} is the trivial one $a_{\mu\nu\rho\lambda}=0$.

Define a Gaussian average of a function $A$ by
\begin{align}
    \bra A \ket \equiv \int d^D u \left( A \ e^{\frac{-1}{2}\eta_{\mu\nu} u^\mu u^\nu} \right).
\end{align}
Taking the average of eq.~\eqref{eq:G_primary_constraint} over $u$, all explicit dependence on $u^\mu$ can be expressed in terms of $\bra u^\mu u^\nu \ket \propto \eta^{\mu\nu}$.
The averaged equation becomes
\begin{align}
    a_{\mu\rho\nu\lambda} \eta^{\mu\nu} v^\rho v^\lambda v^2 + a_{\mu\rho\nu\lambda} \eta^{\mu\nu} v^\rho v^\lambda v^2 = 2 a_{\mu\rho\nu\lambda} \eta^{\mu\nu} v^\rho v^\lambda v^2 = 0.
\end{align}
Therefore, $a_{\mu\rho\nu\lambda}$ must vanish.
There is no nontrivial areal-metric deformation \eqref{operator-a} that preserves the conformal symmetry by itself.

We emphasize that this argument only applies in the perturbative regime around the flat background in the absence of other background fields.

%%%%%%%%%%%%%%%%%%%%%%%%%%%%%%%%%%%%%%%%%%%%%%%%%%%%%%%%%%%%%%%%%%%%%%%
%%%%%%%%%%%%%%%%%%%%%%%%%%%
\section{Higher Dimensional Extended Objects}
\label{sec:higher_dim}

So far we have focused on string theories, which have two-dimensional worldsheets, and established their classical equivalence on both Riemannian and areal-metric manifolds.
In this section, we extend the analysis to higher-dimensional extended objects and introduce a generalization of an areal metric, namely, a volume metric.

For higher-dimensional extended objects in a Riemannian spacetime, the Nambu-Goto action is naturally generalized to
\begin{align}\label{eq:action_ng_volume}
    \mathcal{L}_{NG}=T\sqrt\gamma,
\end{align}
where $\gamma$ is the determinant of the induced metric $\gamma_{ab}$ on $d$-dimensional worldvolume.
More explicitly, we have
\begin{equation}\label{eq:gggggg}
\begin{split}
    \gamma 
    &=\frac{1}{d!} \epsilon^{a_1a_2\dots a_d}\epsilon^{b_1b_2\dots b_d}\gamma_{a_1b_1}\gamma_{a_2b_2}\dots\gamma_{a_db_d}\\
    &= \frac{1}{d!} g_{\mu_1\nu_1}g_{\mu_2\nu_2}\dots g_{\mu_d\nu_d}\sigma^{\mu_1\mu_2\dots \mu_d}\sigma^{\nu_1\nu_2\dots \nu_d},
\end{split}    
\end{equation}
where $\gamma_{ab}$ is defined by
\begin{align}
    \gamma_{ab} \equiv g_{\mu\nu}(X) \partial_a X^\mu \partial_b X^\nu 
    \qquad (a, b = 0, 1, \cdots, d-1),
    \label{eq:gamma-1}
\end{align} 
and
\begin{align}\label{sigma-def-v}
    \sigma^{\mu_1\mu_2\dots\mu_d}
    &\equiv \epsilon^{a_1a_2\dots a_d}\partial_{a_1}X^{\mu_1}\partial_{a_2}X^{\mu_2}\dots \partial_{a_d}X^{\mu_d},
\end{align}
which can be expressed as the $d$-dimensional Nambu bracket $\left\{X^{\mu_1},X^{\mu_2},\dots,X^{\mu_d}\right\}$.
The Lagrangian~\eqref{eq:action_ng_volume} can thus be rewritten as
\begin{align}\label{eq:nambu_goto_d_volume}
    \mathcal{L}_{NG}=T\sqrt{\left\{X^{\mu_1},\dots,X^{\mu_d}\right\}^2}.
\end{align}
We define the Schild Lagrangian for higher-dimensional objects as
\begin{align}
    \mathcal{L}_{S}=a{\left\{X^{\mu_1},\dots,X^{\mu_d}\right\}^2}.
\end{align}
We shall consider in this section these actions on Riemannian manifolds as well as their generalizations to manifolds equipped with volume metrics.

Similar to the definition of the areal metric~\eqref{eq:areal_metric_def}, we define the volume metric through the volume element $d(vol)$ as
\begin{align}
    d(vol)^2 \equiv V_{\left[\mu_1\mu_2\dots\mu_d\right]\left[\nu_1\nu_2\dots\nu_d\right]} \left(dX^{\mu_1}\wedge dX^{\mu_2}\wedge \dots \wedge dX^{\mu_d}\right) \otimes_{sym} \left(dX^{\nu_1}\wedge dX^{\nu_2}\wedge \dots \wedge dX^{\nu_d}\right),
\end{align}
where $V_{\left[\mu_1\mu_2\dots\mu_d\right]\left[\nu_1\nu_2\dots\nu_d\right]}$ is the volume metric.

We will extend our discussion of the worldsheet metric in string actions to higher-dimensional extended objects on Riemannian manifolds in Sec.~\ref{sec:worldvolume_action}.
In Sec.~\ref{sec:volumemetric_generalization}, we define the volume-metric Nambu-Goto action and the volume-metric Schild actions, and discuss how they are equivalent.
In Sec.~\ref{sec:VPD_theorem}, we prove a theorem about the VPD symmetry, give some examples, and then discuss a common feature of the generalizations of the Nambu-Goto action.

\subsection{Worldvolume Metric for Higher Dimensional Extended Objects on Riemannian Manifolds} \label{sec:worldvolume_action}

We now extend the discussion in Sec.~\ref{sec:general_worldsheet_actions} to $d$-dimensional worldvolumes.
As in Sec.~\ref{sec:general_worldsheet_actions}, we denote $\Gamma$ and $H$ to be the matrix form of induced metric $\gamma_{ab}$ and worldvolume metric $h_{ab}$, and they are now of size $d\times d$.
The VPD-invariant Lagrangian density has the form
\begin{align}\label{eq:generalL}
    \mathcal{L}=f\left( \sqrt{\gamma}, a_1, \dots,a_d \right)
\end{align}
where $a_1,\dots,a_d$ are the invariants built from the matrix $\Gamma H^{-1}$, which are defined as $a_k=\text{tr}\left(\left(\Gamma H^{-1}\right)^k\right)$ with $k=1,\dots,d$.

Using the identity
\begin{align}
    \frac{\partial a_k}{\partial h^{ij}} = k \ \gamma_{ja_1}h^{a_1b_2}\ \gamma_{b_2a_2}h^{a_2b_3} \dots \gamma_{b_ki}
\end{align}
where $h^{ab}$ is the inverse of $h_{ab}$, the equation of motion for $h^{ij}$ reads
\begin{align}
    \frac{\partial \mathcal{L}}{\partial h^{ij}}=\sum_{k=1}^d k \, f_k \, \gamma_{ja_1}h^{a_1b_2}\, \gamma_{b_2a_2}h^{a_2b_3} \dots \gamma_{b_ki}\equiv 0
\end{align}
where $f_k\equiv \partial f/\partial a_k$.
After multiplying $h^{i\, b_{k+1}}$ on both sides of the last equality, we get
\begin{align}\label{eq:gen_case}
    0=\sum_{k=1}^d k \, f_k \, A^k,
\end{align}
where $A=\Gamma H^{-1}$.
Since the solution to this equation depends on the exact form of $f$, we cannot solve it explicitly.
However, physically speaking, since the induced metric should be full rank for extended objects, the matrix $A$ is invertible generically.
This implies three things.

\begin{enumerate}
    \item In the case of $d=2$, after multiplying $A^{-1}$, eq.~\eqref{eq:gen_case} becomes
    \begin{align}
        0=f_1\, I+2f_2A,
    \end{align}
    where $I$ is the two-dimensional identity.
    This equation implies $A=\frac{-f_1}{2f_2}I$, meaning that the eigenvalues of $A$ must be identical.
    Also, this expression aligns with the result in Sec.~\ref{sec:general_worldsheet_actions}, see eq.~\eqref{eq:h_eom}, and all the discussions in previous sections apply.

    \item The case of $d\geq3$ is totally different from that of $d=2$.
    In such a case, we have more than two terms in the equation, and multiplying $A^{-1}$ does not help to solve the equation.
    Instead, if we assume that the matrix $A$ is diagonalizable, meaning
    \begin{align}
        A=P\Lambda P^{-1},
    \end{align}
    where $\Lambda$ is a diagonal matrix and $P$ is some invertible matrix.
    The $d$ eigenvalues $\lambda_i$ $(i=1,\dots,d)$ in $\Lambda$ satisfy eq.~\eqref{eq:gen_case}, that is,
    \begin{align}\label{eq:char_eq_gen}
        0=\sum_{k=1}^d k\,f_k\left(\sqrt{\gamma},a_1\left(\lambda_1,\dots,\lambda_d\right),\dots,a_d\left(\lambda_1,\dots,\lambda_d\right)\right)\,\lambda_i^k,
    \end{align}
    which is a system of $d$ coupled equations of $\lambda_i$'s. 
    Generically, we have a discrete solution set for the system $S=\left\{\left(\lambda_1^{(\alpha)}\left(\sqrt{\gamma}\right),\dots,\lambda_d^{(\alpha)}\left(\sqrt{\gamma}\right)\right)\right\}$ where $\alpha$ labels the solutions.
    When we substitute them into our Lagrangian \eqref{eq:generalL}, we can have multiple possible Lagrangians 
    \begin{align}
        \mathcal{L}^{(\alpha)} = f\left(\sqrt{\gamma},a_1^{(\alpha)}\left(\sqrt{\gamma}\right),\dots,a_d^{(\alpha)}\left(\sqrt{\gamma}\right)\right),
    \end{align}
    where $a_n^{(\alpha)}\left(\sqrt{\gamma}\right)=\sum_{k=1}^d ({\lambda_k^{(\alpha)}}\left(\sqrt{\gamma}\right))^n$ for $n=1,\dots,d$.
    For any $\alpha$, the resulting Lagrangians are functions of $\sqrt{\gamma}$.
\end{enumerate}

At any rate, after eliminating the auxiliary worldvolume metric $h_{ab}$, the resulting Lagrangian is some function of $\sqrt{\gamma}$.
In Sec.~\ref{sec:equiv_gen_the_Schild_ng}, we will show that all of them are equivalent to the Nambu-Goto action with a dynamically determined tension.
Before that, we generalize the discussion above to manifolds equipped with volume metric in the next section.

\subsection{Higher Dimensional Extended Objects on Volume-Metric Manifolds}\label{sec:volumemetric_generalization}

As in the areal metric case, we can generalize $g_{\mu_1\nu_1}g_{\mu_2\nu_2}\dots g_{\mu_d\nu_d}$ in eq.~\eqref{eq:gggggg} to the $d$-dimensional volume metric $V_{[\mu_1\mu_2\dots\mu_d][\nu_1\nu_2\dots\nu_d]}$, where indices in the bracket $[\cdots]$ are totally anti-symmetric among themselves, and the exchange of the two $d$-tuples $([\mu_1\mu_2\dots\mu_d] \leftrightarrow [\nu_1\nu_2\dots\nu_d])$.
Then, we have the Nambu-Goto action on a volume-metric background as 
\begin{align}
    S^V_{NG}=T\int d^d\xi\sqrt{\gamma_v},
    \label{SVNG}
\end{align}
with
\begin{align}\label{eq:gamma_v}
    \gamma_v \equiv V_{[\mu_1\mu_2\dots\mu_d][\nu_1\nu_2\dots\nu_d]}\sigma^{\mu_1\mu_2\dots\mu_d}\sigma^{\nu_1\nu_2\dots\nu_d}\equiv{\sigma_v}^2
\end{align}
and $\sigma^{\mu_1\mu_2\dots\mu_d}$ is defined by eq.~\eqref{sigma-def-v}.
The counterpart of the generalized Schild action on a volume-metric background is given by
\begin{align} \label{eq:sigma_volume_gen}
    \mathcal{L}^V_{GS}=F(\gamma_v)
\end{align}
with an arbitrary function $F$.

Below, we show its equivalence to the volume-metric Schild theory
\begin{align}
    \mathcal{L}^V_{S} \equiv \gamma_v = \sigma_v^2.
\end{align}
The equation of motion with respect to $X^{\mu_1}$ for eq.~\eqref{eq:sigma_volume_gen} is
\begin{equation}
\begin{split}
    0
    =&\epsilon^{a_1a_2\dots a_d}\partial_{a_1}\left\{\left(\partial_{a_2}X^{\mu_2}\dots\partial_{a_d}X^{\mu_d}\right) F'\left(\gamma_v\right)\sigma_{\mu_1\mu_2\dots\mu_d}\right\}\\
    =&\epsilon^{a_1a_2\dots a_d}\left(\partial_{a_2}X^{\mu_2}\dots\partial_{a_d}X^{\mu_d}\right)\partial_{a_1}\left\{ F'\left(\gamma_v\right)\sigma_{\mu_1\mu_2\dots\mu_d}\right\}\\
    =&\epsilon^{a_1a_2\dots a_d}\left(\partial_{a_2}X^{\mu_2}\dots\partial_{a_d}X^{\mu_d}\right) \times
    \\
    &\times \left[
    \left(\partial_{a_1}{\sigma_v}^2\right)F''\left(\gamma_v\right)\sigma_{\mu_1\mu_2\dots\mu_d}+F'\left(\gamma_v\right)\left(\partial_{a_1}\sigma_{\mu_1\mu_2\dots\mu_d}\right)
    \right]
\label{eq:eom_volume}
\end{split}
\end{equation}
where $F'$ and $F''$ denote derivatives with respect to $\sigma_v^2$.

Multiplying eq.~\eqref{eq:eom_volume} by $\partial_bX^{\mu_1}$ and using the $d$-dimensional analogue of the identity~\eqref{eq:antisym_sum}
\begin{align}
    \partial_{a_1}X^{[\mu_1}\partial_{a_2}X^{\mu_2}\dots \partial_{a_d}X^{\mu_d]}=\epsilon_{a_1a_2\dots a_d}\sigma^{\mu_1\mu_2\dots\mu_d},
\end{align}
where the indices $\mu_1, \cdots, \mu_d$ on the left-hand side are totally anti-symmetrized, we obtain
\begin{align}\label{eq:gen_eom_vv}
    0=\left(\partial_b\gamma_v\right)\left(2\gamma_v F''\left(\gamma_v\right)+F'\left(\gamma_v\right)\right).
\end{align}

For the volume-metric Nambu-Goto action~\eqref{SVNG}, $F(\gamma_v) = T \sqrt{\gamma_v}$, the equation above is trivially satisfied for any $\gamma_v$ because $\left(2\gamma_v F''\left(\gamma_v\right)+F'\left(\gamma_v\right)\right) = 0$ identically.
(The equivalence of the Nambu-Goto action with the Schild action will be proved in Sec.~\ref{sec:equiv_gen_the_Schild_ng}.)

For a generic $F$, since $\left(2\gamma_v F''\left(\gamma_v\right)+F'\left(\gamma_v\right)\right)$ does not vanish identically, the equation of motion implies
\begin{align}\label{eq:sv2=c}
\left(\partial_b\gamma_v\right) =
\partial_b\left(\sigma_v^2\right) = 0.
\end{align}
(See Appendix~\ref{sec:discrete_sigma} for a detailed proof.) 
Thus $\gamma_v = \sigma_v^2$ is constant on the worldvolume, meaning that $F'(\gamma_v)$ is also a constant on the worldvolume.
This allows us to extract $F'\left(\gamma_v\right)$ from the first line of eq.~\eqref{eq:eom_volume}.
It becomes
\begin{align}
    0 = \epsilon^{a_1a_2\dots a_d} \, \partial_{a_1} \left\{\left(\partial_{a_2}X^{\mu_2}\dots\partial_{a_d}X^{\mu_d}\right) \sigma_{\mu_1\mu_2\dots\mu_d}\right\},
\end{align}
which is nothing but the equation of motion of the volume-metric Schild theory as long as $F$ is not independent of $\gamma_v$.

Here, we have proved the equivalence between the volume-metric generalized Schild action and the volume-metric Schild action. 
This contains the Riemannian manifolds as special cases, since the line metric is a special case of the volume metric.
In Sec.~\ref{sec:equiv_gen_the_Schild_ng}, we provide a proof to show the equivalence between the (volume-metric) generalized Schild action and the (volume-metric) Nambu-Goto action.

\subsection{A Theorem on VPD symmetry}\label{sec:VPD_theorem}

\begin{theorem}
\label{Thm}
For an arbitrary action of the form
\begin{align}
    S=S[\phi; \rho],
\end{align} 
where $\phi$ represents any collection of dynamical scalar, vector, or tensor fields, while $\rho$ is a background scalar density field,
if the action is invariant under simultaneous diffeomorphism transformations of both $\phi$ and $\rho$,
the quantity $\frac{\delta S}{\delta \rho}$ is a spacetime constant when $\phi$ satisfies the equations of motion.
\end{theorem}

\begin{proof}
The proof goes as follows.
By definition, the scalar density $\rho$ transforms as 
\begin{align}
\delta_v \rho = \partial_\mu \left( \rho v^\mu \right)
\end{align}
under a diffeomorphism $v^{\mu}$.
On the other hand, the action $S$ transforms as
\begin{align}
    0 = \int d^D x\left( \frac{\delta S}{\delta \phi\left(x\right)} \delta_v \phi\left(x\right)+\frac{\delta S}{\delta \rho\left(x\right)}\partial_\mu \left( v^\mu \rho \right) \right).
\end{align}
By taking $\phi\left(x\right)$ on shell, the first term vanishes.
Via integration by parts, the 2nd term becomes
\begin{align}
    0 = \int d^D x \left( - \partial_\mu\left( \frac{\delta S}{\delta \rho(x)} \right) v^\mu \rho \right),
\end{align}
which implies $\partial_\mu\left( \frac{\delta S}{\delta \rho\left(x\right)} \right) = 0$, and therefore $\frac{\delta S}{\delta \rho\left(x\right)}$ is a spacetime constant. 
\end{proof}

\subsubsection{Example: Unimodular Gravity and Cosmological Constant}
\label{sec:ex1}

It is well known that the cosmological constant in unimodular gravity arises as an initial condition of the dynamical equations of motion~\cite{Anderson:1971pn,Henneaux:1989zc}.
This can be viewed as a result of the application of Theorem~\ref{Thm}.

Consider the action
\begin{align}
    S[\phi; \rho] = \int d^Dx\left(\sqrt{g}R + \lambda \left(\sqrt{g}-\rho\right)\right),
\end{align}
where $\phi=\{g_{\mu\nu}, \lambda\}$ are dynamical fields and $\rho$ is the background scalar density field.
It becomes the action for unimodular gravity when the auxiliary field $\lambda$ is integrated out.

By definition, both the metric $g_{\mu\nu}$ and the scalar $\lambda$ are dynamical fields.
In particular, the field $\lambda$ can be any spacetime function.
The field $\rho$ is, on the other hand, a scalar density of a given background.

Theorem \ref{Thm} states that, since $\delta S/\delta \rho = - \lambda$, $\lambda$ is forced to be a constant, i.e.
\begin{align}
    \lambda = C,
\end{align}
when the metric $g_{\mu\nu}$ satisfies the equations of motion.
This makes $\lambda$ a cosmological constant, although its value is to be determined by the initial condition.

This example suggests that the cosmological constant may arise as a result of the equations of motion when the theory respects VPD; we do not need to put in by hand a constant in the theory from the beginning.

\subsubsection{Universal Proof for Constant Area/Volume Density}
\label{sec:equiv_gen_the_Schild_ng}

From our experience in previous examples, we notice that the key to showing the equivalence between the generalized Schild and Nambu-Goto actions is the constant area or volume density $\sigma^2$, e.g. \eqref{eq:sigma_constant}, \eqref{eq:s2=c}, \eqref{eq:areal_G_sigma_sigma}, \eqref{eq:constant_gamma_4} and \eqref{eq:sv2=c}.
Below, we provide a universal proof using Theorem~\ref{Thm} to reproduce these results.
This applies to extended objects of any dimensions on manifolds with any metrics.

Starting with a general form of action on a $d$-dimensional worldvolume,
\begin{align}\label{eq:vpd_action_density}
    \mathcal{S}=\int d^d\xi \left( \rho F\left(\rho^{-1} \sqrt\gamma\right) \right)
\end{align}
where $\rho$ is the background scalar density, and $\sqrt\gamma$ is the induced volume form for arbitrary generalized metrics.
Theorem~\ref{Thm} demands
\begin{align}\label{eq:rho_arbitrary_F}
    C= F\left(\rho^{-1}\sqrt\gamma\right)-\frac{\sqrt\gamma}{\rho}F'\left(\rho^{-1}\sqrt\gamma\right)
\end{align}
where $C$ is some worldvolume constant.\footnote{Since theorem \ref{Thm} holds on worldvolumes as well, we choose to make $C$ worldvolume constant instead of spacetime constant.}
Upon fixing the background $\rho$ to $1$, the equation becomes\footnote{{This is the same procedure as in the ordinary Polyakov theory \cite{Becker:2006dvp}.}
That is, after gauge fixing the auxiliary worldsheet metric $h_{ab}$, the equation of motion for $h_{ab}$ becomes the Virasoro constraint, and the action becomes free bosons on two-dimensional flat space.}
\begin{align}\label{eq:rho1F}
    C=F\left(\sqrt\gamma\right)-\sqrt\gamma F'\left(\sqrt\gamma\right).
\end{align}
% To solve this equation, we define $G\left(\sqrt{\gamma}\right) \equiv F\left(\sqrt{\gamma}\right)-C$.
% Then the equation above becomes
% \begin{align}
%     \frac{\partial}{\partial \sqrt{\gamma}}\left(\frac{G\left(\sqrt{\gamma}\right)}{\sqrt{\gamma}}\right)=0,
%     \label{eq:diff_xi_a_G}
% \end{align}
% implying that $G\left(\sqrt{\gamma}\right)/\sqrt{\gamma}\equiv A$ is a worldvolume constant.
% After substituting this form back to eq.~\eqref{eq:vpd_action_density} ($\rho = 1$ by gauge fixing), we get
% \begin{align}\label{eq:vpd_to_ng}
%     \mathcal{S}=\int d^d\xi\left(C+A\sqrt\gamma \right).
% \end{align}
% Since the first term is merely a constant, we can simply ignore it. 
% The resulting action is the Nambu-Goto action, and it is clear that the worldvolume constant $A$ is the tension on the worldvolume.
For generic $F$, \eqref{eq:rho1F} implies that $\sqrt\gamma$ is a constant on the worldvolume.
Therefore, the Nambu-Goto and the Schild actions\footnote{For $d\geq 3$, the Schild action is defined by the Nambu bracket squared.} are equivalent, and the constant tension is dynamically determined.
Note that this result holds not only for the Riemannian metric, but also for the areal metric and volume metrics.

%%%%%%%%%%%%%%%%%%%%%%%%%%%%%%%%%%%%%%%%%%%%%%%%%%%%%%%%%%%%%%%%%%%%%%%%%%%%%%%%%%%%%%%%%%%%%%%%%%%%

\section{Conclusion} \label{sec:conclusion}

In the first part of the paper (Secs.~\ref{Riemann} and~\ref{sec:the Nambu-Goto_the_Schild_equiv}), we have studied general string actions on Riemannian manifolds that are functions of worldsheet metric $h_{ab}$ and induced metric $\gamma_{ab}$.
This includes three well-known string actions: the Schild, the Nambu-Goto, and the Polyakov actions.
They have different symmetries: VPD, Diff, and Diff-Weyl, respectively.
Hence, we consider the general forms of actions that respect these symmetries and show their classical equivalence. 
Also, we discuss how tension can be defined from the Schild action by matching its stress-energy tensor to the Nambu-Goto's.

In the second part (Secs.~\ref{sec:algebraic_areal_bosonic} --~\ref{sec:anomalous_polyakov}), we have extended the analysis from Riemannian manifolds to areal-metric manifolds.
For VPD and diffeomorphisms, we have constructed the corresponding areal-metric analogues of the Schild action and the Nambu-Goto action\footnote{Note that the only action that is diffeomorphism and Weyl invariant is the areal Nambu-Goto action~\eqref{eq:action_GNG} when there is only auxiliary worldsheet areal metric $\mathcal{H}_{abcd}$ but no auxiliary worldsheet line metric $h_{ab}$.} and have shown that they are classically equivalent.
We also have shown that the areal-metric type perturbation around the ordinary Polyakov action cannot describe critical strings. 

In Sec.~\ref{sec:higher_dim}, we have generalized the discussion to $d$-dimensional worldvolumes on manifolds with a volume metric, the higher-dimensional analogue of an areal metric.
Using a general theorem on VPD in Sec.~\ref{sec:VPD_theorem}, we show that on $d$-dimensional Riemannian and volume-metric manifolds, the Schild actions are classically equivalent to the Nambu-Goto action with the tension determined dynamically as an integration constant. 
As long as we regard our universe as a single connected worldvolume, this implies that the Schild theories are classically equivalent to the Nambu-Goto theory in any dimension and on any Riemannian and volume-metric backgrounds.

Through our argument, there are only two natural directions for genuine generalizations.
First, one considers actions with explicit derivatives of $h_{ab}$ and $\gamma_{ab}$.
In such case, the analysis above can be regarded as their low-energy limit.
We did not discuss this case here, but we expect they will provide a nontrivial modification to the theory.
Second, as we discussed in Sec.~\ref{sec:anomalous_polyakov}, when the (line) metric is generalized to an areal metric, the modification is nontrivial in that the modified theory can no longer describe critical strings unless further modifications are introduced.
Hence, it would be interesting to investigate systematically whether consistent (possibly non-perturbative) string or brane theories can be defined on such generalized metric backgrounds, e.g., through AdS/CFT correspondence~\cite{Bhattacharya:2025hag}.

%%%%%%%%%%%%%%%%%%%%%%%%%%%%%%%%%%%%%%%%%%%%%%%%%%%%%%%%%%%%%%%%%%%%%%%
\section*{ACKNOWLEDGMENTS}
P. M. H. is supported in part by the National Science and Technology Council, R.O.C. (NSTC 113-2112-M-002 -040 -MY2), and by National Taiwan University.
H.K. is partially supported by JSPS (Grants-in-Aid for Scientific Research Grants No. 20K03970), and by National Taiwan University.
H.K. also thanks Prof. Shin-Nan Yang and his family for their kind support through the Chin-Yu chair professorship.
H. L. is supported in part by the Ministry of Science and Technology grant (112-2112-M-002-024-MY3).

\appendix

\section{Details on solving $\partial_a\left(\sigma^2\right)\left(2\sigma^2 F''(\sigma^2)+F'(\sigma^2)\right)=0$}\label{sec:discrete_sigma}

In this appendix, we discuss in detail how one solves equations of the type
\begin{align}
    \partial_a\left(\sigma^2\right)\left(2\sigma^2 F''(\sigma^2)+F'(\sigma^2)\right)=0
\end{align}
where $\sigma^{2}$ is a function on the worldsheet/worldvolume and $\partial_{a}$ is the derivative on it.
% where $\partial_a$ is the derivative on the worldsheet/worldvolume\footnote{Here, we do not specifically distinguish worldsheet and worldvolume, and below we call both of them worldvolume for simplicity.}, and $\sigma^2$ is defined over manifolds with any type of metric, such as a line metric $g_{\mu\nu}$, an areal metric $G_{\mu\nu\rho\lambda}$ or a volume metric $V_{\mu_1\dots\mu_d,\nu_1\dots\nu_d}$ with $d$ matching the dimension of the manifold.

% For generic $F$, indeed
Generically $\partial_a\sigma^2$ has to vanish. 
However, there may be specific points on the worldvolume on which $2\sigma^2 F''(\sigma^2)+F'(\sigma^2)$ vanishes.
We assume that the solution set $S$ for $\sigma^2$ to $2\sigma^2 F''(\sigma^2)+F'(\sigma^2)=0$ is discrete.
Suppose that there is a point $\xi_*^a$ on the worldvolume that simultaneously satisfies 
\begin{align}
    \partial_a\sigma^2&\neq 0, \\
    \sigma^2&=\sigma^2_*\in S.
\end{align}
Assuming that $\sigma^2$ is smooth function around a sufficiently small neighborhood of $\xi_*^a$, these conditions still remain.
However, this means $\partial_a\sigma^2=0$ around the point $\xi_*^a$.
Hence, when the set of solutions of $2\sigma^2 F''(\sigma^2)+F'(\sigma^2)=0$ is discrete, $\partial_a\sigma^2$ must vanish.

%%%%%%%%%%%%%%%%%%%%%%%%%%%%%%%%%%%%%%%%%%%%%%%%%%%%%%%%%%%%%%%%%%%%%%%%%

\end{document}